\documentclass{article}

\usepackage{amsmath,amsthm,amssymb}
\usepackage{graphicx}
\usepackage{xcolor} % Only used for comments at the moment
\usepackage{url}
\usepackage{hyperref}

\graphicspath{ {./figures/} }

\theoremstyle{plain}
\newtheorem{theorem}{Theorem}
\newtheorem{lemma}[theorem]{Lemma}
\newtheorem{observation}[theorem]{Observation}
\newtheorem{conjecture}[theorem]{Conjecture}

\newcommand{\edgeset}{\mathcal{E}}

\title{Snipperclips: Cutting Tools into Desired Polygons using Themselves\thanks{An extended abstract of this paper appeared in the proceedings of the 29th Canadian Conference on Computational Geometry (CCCG 2017)~\cite{SnipperclipsCCCG}. 
M.~C. was supported by ERC StG 757609.
M.~K. was partially supported by MEXT KAKENHI Nos.~12H00855, and 17K12635. M.-K.~C., M.~R. and A.~v.~R. were supported by JST ERATO Grant Number JPMJER1201, Japan.}}

\author{
Zachary Abel\thanks{Massachusetts Institute of Technology, Cambridge, Massachusetts, USA, \protect\url{{zabel,edemaine,mdemaine,achester,}@mit.edu}}
\and Hugo Akitaya\thanks{University of Massachusetts Lowell, USA, \protect\url{hugo_akitaya@uml.edu}}
\and Man-Kwun Chiu\thanks{Institut f\"ur Informatik, Freie Universit\"at Berlin,
\protect\url{chiumk@zedat.fu-berlin.de}}
\and Erik D. Demaine\footnotemark[2]
\and Martin L. Demaine\footnotemark[2]
\and Adam Hesterberg\footnotemark[2]
\and Matias Korman\thanks{Siemens EDA (formerly Mentor Graphics), OR, USA. \protect\url{matias_korman@mentor.com}.}
\and Jayson Lynch\thanks{University of Waterloo, Ontario, Canada.\protect\url{jayson.lynch@uwaterloo.ca}}
\and Andr\'e van Renssen\thanks{University of Sydney, Sydney, Australia, \protect\url{andre.vanrenssen@sydney.edu.au}}
\and Marcel Roeloffzen\thanks{TU Eindhoven, Eindhoven, the Netherlands, \protect\url{m.j.m.roeloffzen@tue.nl}}
}

\date{}

\begin{document}

\maketitle

\begin{abstract}
We study \emph{Snipperclips}, a computer puzzle game whose objective is to create a target shape with two tools. The tools start as constant-complexity shapes, and each tool can snip (i.e., subtract its current shape from) the other tool. We study the computational problem of, given a target shape represented by a polygonal domain of $n$ vertices, is it possible to create it as one of the tools' shape via a sequence of snip operations? If so, how many snip operations are required? We consider several variants of the problem (such as allowing the tools to be disconnected and/or using an undo operation) and bound the number of operations needed for each of the variants.
\end{abstract}

%%%%%%%%%%%%%%%%%%%%%%%%%%%%%%%%%%%%
\section{Introduction}
%%%%%%%%%%%%%%%%%%%%%%%%%%%%%%%%%%%%
\emph{Snipperclips: Cut It Out, Together!}\ \cite{Snipperclips-wiki}
is a puzzle game developed by SFB Games and published by Nintendo worldwide on March 3, 2017 for their new console, Nintendo Switch. In the game, up to four players cooperate to solve puzzles.
Each player controls a character%
\footnote{The game in fact allows one human to control up to two characters,
  with a button to switch between which character is being controlled.}
whose shape starts as a rectangle in which two corners have been
rounded so that one short side becomes a semicircle.  The main
mechanic of the game is \emph{snipping}: when two such characters
partially overlap, one character can \emph{snip} the other character,
i.e., subtract the current shape of the first character from the
current shape of the latter character; see Figure~\ref{fig:game} (top middle) where the yellow character snips the red character subtracting from it their intersection (which is shown in green).  
In addition, a \emph{reset} operation allows a character to restore its
original shape. Finally, an \emph{undo} operation allows a character
to restore its shape to what it was before the prior snip or reset
operation. A more formal definition of these operations follows in the next section. An unreleased 2015 version of this game,
\emph{Friendshapes} by SFB Games, had the same mechanics, but
supported only up to two players \cite{Friendshapes}. 

\begin{figure}[!b]
\centering
\includegraphics[width=0.32\linewidth]{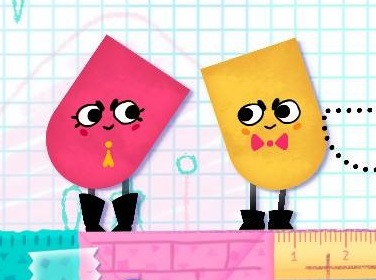}\hfill
\includegraphics[width=0.32\linewidth]{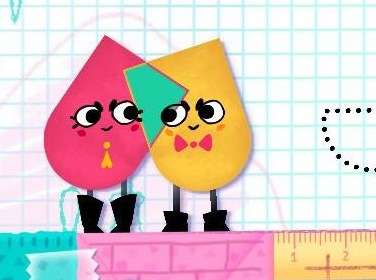}\hfill
\includegraphics[width=0.32\linewidth]{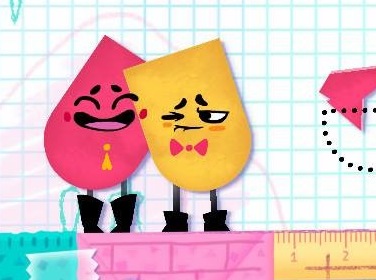}

\smallskip

\includegraphics[width=0.32\linewidth]{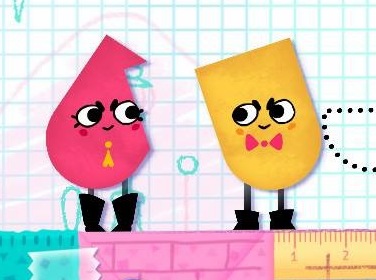}\hfill
\includegraphics[width=0.32\linewidth]{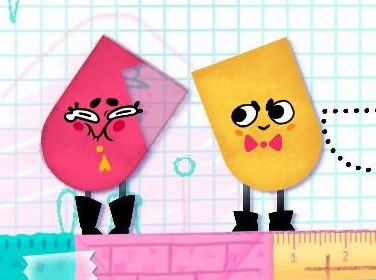}\hfill
\includegraphics[width=0.32\linewidth]{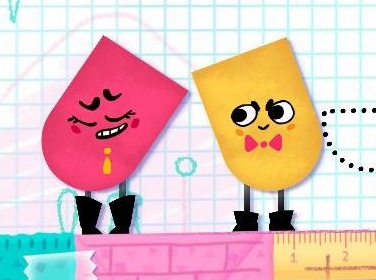}

\smallskip

\includegraphics[width=0.49\linewidth]{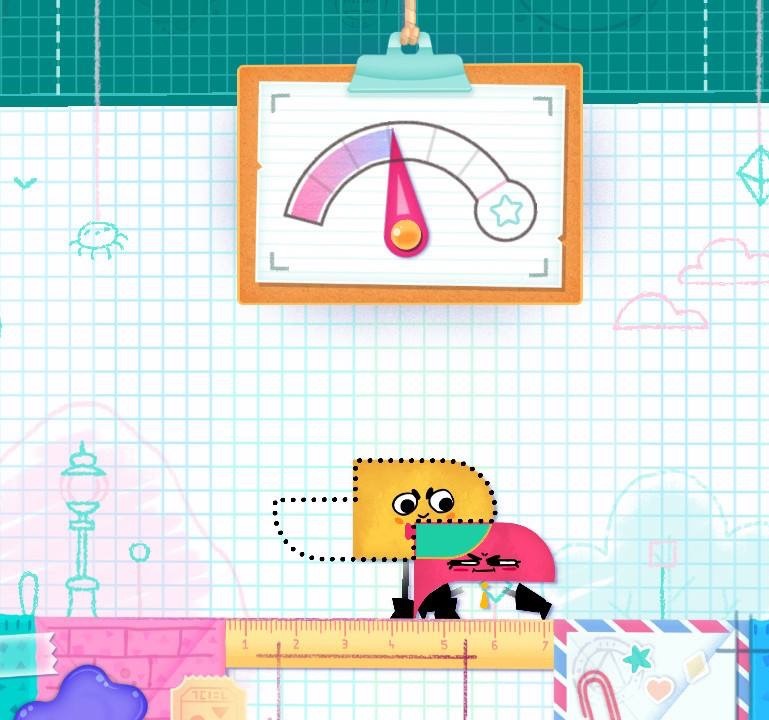}\hfill
\includegraphics[width=0.49\linewidth]{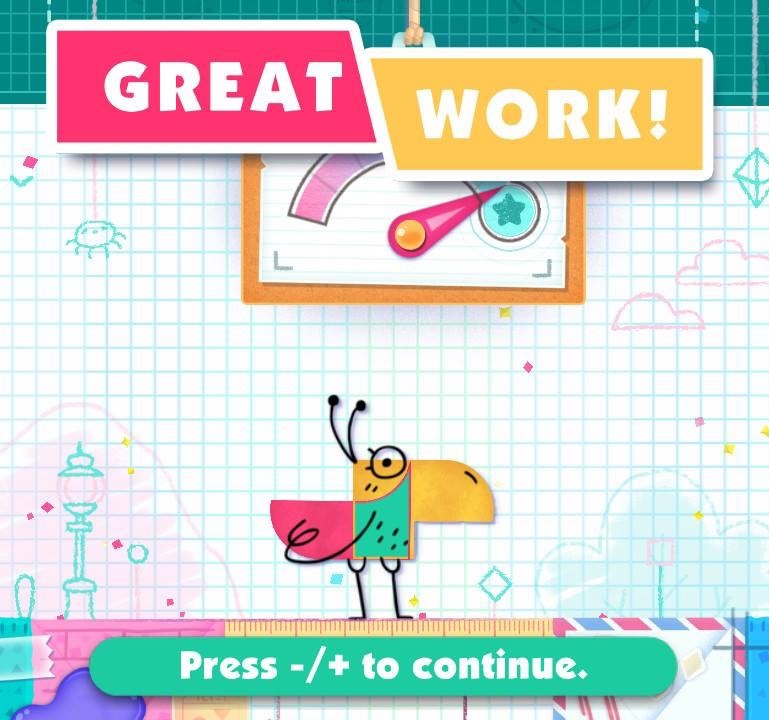}

\caption{Cropped screenshots of Snipperclips: snipping, resetting,
  and solving a Shape Match puzzle.  Sprites copyright SFB/Nintendo
  and included here under Fair Use.}
\label{fig:game}
\end{figure}

Puzzles in Snipperclips have varying goals, but an omnipresent subgoal is to
form one or more players into desired shape(s), so that they can carry out
required actions.  In particular, a core puzzle type (``Shape Match'')
has one target shape which must be (approximately) formed by the union of the
characters' shapes.  
In this paper, we study when
this goal is attainable, and when it is, analyze the minimum number of
operations required.

\section{Problem definition and results}
For the remainder of the paper we consider the case of exactly two characters
or \emph{tools} $\mathcal{T}_1$ and $\mathcal{T}_2$.  For geometric simplicity,
we assume that the initial shape of both tools is a unit square.
Most of the results in this paper work for nice (in particular, fat)
constant-complexity initial shapes, such as the rounded rectangle in
Snipperclips, but would result in a more involved description.

We view each tool as an open set of points that can be rotated and translated freely.%
\footnote{In the actual game, the tools' translations are limited by gravity, jumping, crouching, stretching, standing on each other, etc., though in practice this is not a huge limitation.  Rotation is indeed arbitrary.}
After any rigid transformation, if the two tools have nonempty intersection, we can \emph{snip} (or \emph{cut}) one of them, i.e., remove from one of the tools the closure of the intersection of the two tools (or equivalently, the closure of the other tool, see Figure~\ref{fig:snip}). Note that by removing the closure we preserve the invariant that both tools remain open sets. In addition to the snip operation, we can \emph{reset} a tool, which returns it back to its original unit-square shape.

\begin{figure}[h]
\centering
\includegraphics[width=\linewidth]{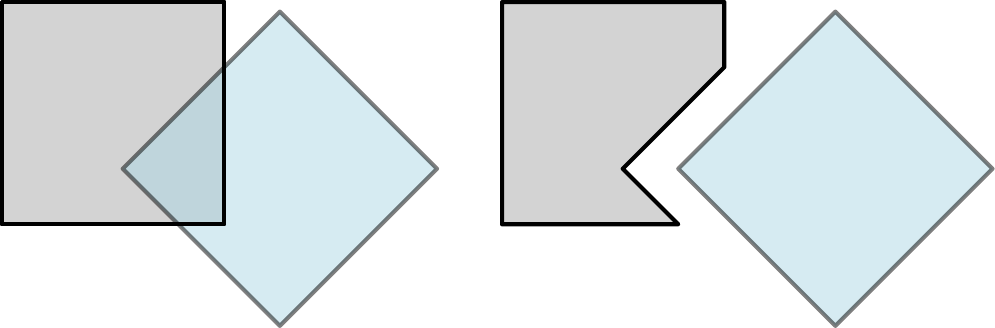}
\caption{By translating and rotating the two tools we can make them partially overlap (left figure). On the right we see the resulting shape of both tools after the snip operation.}
\label{fig:snip}
\end{figure}

After a snip operation, the changed tool could become disconnected. There are two natural variants on the problem of how to deal with disconnection. In the \emph{connected model}, we force each tool to be a single connected component. Thus, if the snip operation disconnects a tool, the user can choose which component to use as the new tool.
In the \emph{disconnected model}, we allow the tool to become disconnected, viewing a tool as a set of points to which we apply rigid transformations and the snip/reset operation.
The Snipperclips game by Nintendo follows the disconnected model, but we find the connected model an interesting alternative to consider.

\label{undo/redo}
The actual game has an additional
\emph{undo/redo} operation, allowing each tool to return into its previous
shape. For example, a heavily cut tool can reset to the square, cut something in the other tool, and use the undo operation to return to its previous cut shape. The game has an undo stack of size~$1$; we consider a more general case in which the stack could have size $0$, $1$ or $2$.

\subsection{Results}
\begin{table}[tb]
\begin{center}
\begin{tabular}{ |c|cc||cc| }
 \hline
 & \multicolumn{2}{c||}{Connected Model} & \multicolumn{2}{|c|}{Disconnected Model} \\ \hline
 Undo stack size & 1 shape & 2 shapes & 1 shape & 2 shapes\\
 \hline
 0 & $O(n)$ & No & $O(n^2)$ & No\\
 1 & $O(n)$ & $O(n+m)$ & $O(n)$ & Yes \\
 2 & $O(n)$ & $O(n+m)$ & $O(n)$ & $O(n+m)$ \\
 \hline
\end{tabular}
\end{center}
\caption{Number of operations required to carve out the target shapes of $n$ and $m$ vertices, respectively. A cell entry with ``No" means that it is not always possible to do whereas ``Yes" means it is possible (but the number of operations needed is not bounded by any function of $n$ or $m$).}
\label{tb:undo}
\end{table}
Given two target shapes $P_1$ and $P_2$, we would like to find a sequence of snip/reset operations that transform tool $\mathcal{T}_1$ into $P_1$ and at the same time transform $\mathcal{T}_2$ into~$P_2$. Because our initial shape is polygonal, and we allow only finitely many snips, the target shapes $P_1$ and $P_2$ must be polygonal domains of $n$ and $m$ vertices, respectively. Whenever possible, our aim is to transform the tools into the desired shapes using as few snip and reset operations as possible. Specifically, our aim is for the number of snip and reset operations to depend only on~$n$ and~$m$ (and not depend on other parameters such as the feature size of the target shape).

In Section~\ref{sec_neg}, we prove some lower bound results.
First we show in Section~\ref{sec_lb_impossibility} that, without an undo operation,
it is not always possible to cut both tools into the desired shape,
even when $P_1 = P_2$.
Then we show lower bounds on the number of snips/undo/redo/reset operations
required to make a single target shape $P_1$.
For the connected model, Section~\ref{sec_lb_connected}
proves an easy $\Omega(n)$ lower bound.
For the disconnected model, Section~\ref{sec_lb_disconnected}
gives a family of shapes that need $\Omega(n)$ operations to carve in a natural 1D model,
and gives a lower bound of $\Omega(\log n)$ for all shapes in the 2D model.

On the positive side, we first consider the problem without the undo operation
in Section~\ref{sec_noundo}.
We give linear and quadratic constructive algorithms to carve a single shape $P_1$ in both the connected and disconnected models, respectively.

In Section~\ref{sec_undo} we introduce the undo operation. We first show that even a stack of one undo allows us to cut both tools into the target shapes, although the number of snip operations is unbounded if we use the disconnected model. We then show that by increasing the undo stack size, we can reduce the number of operations needed to linear. A summarizing table of the number of snips needed depending on the model is shown in Table~\ref{tb:undo}.

\subsection{Related Work}

Computational geometry has considered a variety of problems related to
cutting out a desired shape using a tool such as
circular saw \cite{CircularSaw}, hot wire \cite{Jaromczyk-Kowaluk-2000},
and glass cutting \cite{Overmars-Welzl-1985,Pach-Tardos-2000}.
The Snipperclips model is unusual in that the tools are themselves the
material manipulated by the tools.  This type of model arises in real-world
manufacturing, for example, when using physical objects to guide the
cutting/stamping of other objects---a feature supported by the popular new
\emph{Glowforge} laser cutter \cite{Glowforge} via a camera system.

Our problem can also be seen as finding the optimal Constructive Solid
Geometry (CSG)~\cite{foley1996computer} expression tree, where leaves represent base shapes
(in our model, rectangles), internal nodes represent shape subtraction, and
the root should evaluate to the target shape,
such that the tree can be evaluated using only two registers.
Applegate et al.~\cite{Applegate-2007} studied
a rectilinear version of this problem (with union and subtraction, and a
different register limitation).

\section{Lower Bounds}\label{sec_neg}
In this section, we first prove that some pairs of target shapes cannot be realized
in both tools simultaneously, using only snip and reset operations.
Then we focus on achieving only one target shape.
In the connected model, we give a linear lower bound (with respect to the number $n$ of vertices of the target shape) on the number of operations to construct the target shape.
In the disconnected model, we give a logarithmic lower bound, and give a linear lower
bound in a natural 1D version of Snipperclips.

\subsection{Impossibility} \label{sec_lb_impossibility}
We begin with the intuitive observation that not all combinations of target shapes can be constructed when restricted to the snip and reset operations.

\begin{observation}\label{obs_notpos}
In both the connected and disconnected models,
there is a target shape that cannot be realized by both tools at the same time using only snip and reset operations.
\end{observation}
\begin{proof}
Consider the target shape shown in Figure~\ref{fig:impos}: a unit square in which we have removed a very thin rectangle, creating a sort of thick ``U''. First observe that, if we perform no resets, neither tool has space to spare to construct a thin auxiliary shape to carve out the rectangular gap of the other tool. Thus, after we have completed carving one tool, the other one would need to reset. This implies that we cannot have the target shape in both tools at the same time.

\begin{figure}[ht]
\centering
\includegraphics{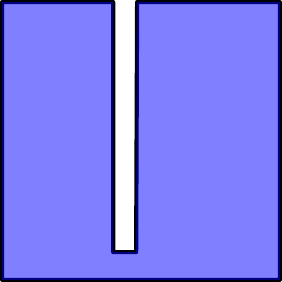}
\caption{A target shape that cannot be realized by both tools at the same time.}
\label{fig:impos}
\end{figure}

Now assume that we can transform both tools into the target shape by performing a sequence of snips and resets. Consider the state of the tools just after the last reset operation. One of the two shapes is the unit square and thus we still need to remove the thin hole using the other shape. However, because no more resets are executed, the other tool is currently and must remain a superset of the target shape.
In particular, it can differ from the square only in the thin hole, so it does not have any thin portions that can carve out the hole of the other tool.

Because the above argument is based solely on the shape of the figure, it holds in both the connected and disconnected model.
\end{proof}

\subsection{Connected Model} \label{sec_lb_connected}

Next it is easy to see that in the connected model a target shape with $\Theta(n)$ holes requires $\Omega(n)$ operations.

\begin{theorem}\label{thm_connected}
	There are target shapes that require $\Omega(n)$ operations (snip, reset, undo and redo) to construct in the connected model.
\end{theorem}

\begin{proof}
	Consider the target shape to be a square with $n/3$ triangular holes. Since we consider the connected model, the cutting tool created by any operations is connected and it can only carve out one hole at a time. 
\end{proof}

\subsection{Disconnected Model} \label{sec_lb_disconnected}

In the disconnected model, we conjecture that most shapes require
$\Omega(n)$ snip operations to produce
(see Conjecture~\ref{conj:2D disconnected}), but such a proof or explicit
shape remains elusive.
The challenge is that a cutting tool may be reused many times,
which for some shapes leads to an exponential speedup.
Indeed, we prove in Theorem~\ref{thm:disconnected log}
that \emph{every} shape requires $\Omega(\log n)$ snips.
As a step toward a linear lower bound, we prove that a natural 1D version
of the disconnected Snipperclips model has a linear lower bound.

Define the \emph{disconnected 1D Snipperclips} model
(with arbitrarily many tools) as follows.
A \emph{1D tool} is a disjoint set of intervals in $\mathbb R$.
A \emph{1D snip} operation takes a translation of one tool,
optionally reflects it around the origin, and subtracts it from another
tool, producing a new tool.

The main difference with the disconnected model that we consider is that we allow for arbitrarily many tools. Alternatively, this is can be done with two tools if you can recall any shape that has been created in the past (i.e., having infinitely many undo, redo, and reset operations).

\begin{theorem} \label{thm:1D disconnected}
For $M$ a positive integer, consider the set of all 1D tools consisting of $n$ disjoint intervals
  having integer endpoints between $0$ and~$M$.
  For all positive integers $n$ and all $\epsilon \in (0,1)$, for all sufficiently large $M$, almost all such tools (at least a $1-\epsilon$ fraction of them) require
  at least $2n$ 1D snip operations
  to build from a single 1D tool consisting of a single interval.
\end{theorem}

\begin{proof}
Starting from $k=1$, the $k$th snip operation is determined by:
\begin{enumerate}
\item A choice of the $k+1$ existing tools for the cutting tool $T$; 
\item A choice of the $k+1$ existing tools for the cut tool $U$;
\item An offset $x_k$ of $U$ relative to $T$.
\end{enumerate}

If $T$ has interval endpoints $t_0$, $t_1$, \ldots~and $U$ has interval endpoints $u_0$, $u_1$, \ldots, then each interval endpoint of the tool created by the $k$th operation is either $t_j$ or $x_k+u_j$. The first tools have interval endpoints $0$ and $x_0 = M$ (the board width), so by induction on $k$, each interval endpoint of the tool created by the $k$th operation is of the form $\sum_{i \in I}x_i$ for some $I \subset \{0, \ldots, k\}$. Therefore, if we make, with $k < 2n-1$ operations, a tool with endpoints $y_0$, \ldots, $y_{2n-1}$, then there is a $(2n-1) \times (k+1)$ 0-1 matrix $A$ such that 
$A
\begin{bmatrix}
x_0\\
\vdots\\
x_k
\end{bmatrix}
=
\begin{bmatrix}
y_0\\
\vdots\\
y_{2n-1}
\end{bmatrix}.
$

If this matrix has rank $rk(A) < k+1$, set $k+1 - rk(A)$ of the $x_i$ to be 0 such that it still has a solution, and choose $rk(A)$ of the $y_i$ such that the $rk(A) \times rk(A)$ square matrix $B$ formed by restricting to the rows corresponding to nonzero $x_i$ and columns corresponding to those $y_i$ is full-rank. $\det(B)$ is a sum, over $rk(A)!$ permutations $\sigma$, of a product of entries of $A$ (or its negative). The entries of $A$ are 0 or 1, so $0 < |\det(B)| \le rk(A)! \le (k+1)!$. All the $y_i$ are integers, so for all $i$, $\det(B)x_i$ is an integer, since we have that $Bx = y$, so $x = B^{-1}y$, and $B^{-1}$ is $1/\det(B)$ times the cofactor matrix of B (which has only integer entries).

Therefore, the $k$th snip operation has at most $(k+1)^2$ choices for the cutting tools and $(k+1)!M < (k+1)^{k+1}M$ choices for the offset $x_k$, so the number of choices for operations up to the $(k-1)$st is at most 
$M^{k-1} k^{k^2}$. On the other hand, the number of 1D tools consisting of $n$ intervals with integer endpoints $y_0$, \ldots, $y_{2n-1}$ between $0$ and $M$ is 
$\binom{M}{2n} > (M-2n)^{2n} > 
(\frac{M}{2})^{2n}$.
If $k-1 < 2n$, then the total number of integer-endpoint tools with $n$ intervals is asymptotically (for large $M$) at most $O(M^{-1})$ times the 
number of integer-endpoint tools we can build in $k-1$ steps, so almost all integer-endpoint tools with $n$ intervals require at least $2n$ steps, as claimed.
In particular, if $\epsilon \in (0,1)$ and $M > 2^{2n}(2n)^{(2n)^2}\epsilon^{-1}$, then at most an $\epsilon$ fraction of
such tools can be built in fewer than $2n$ snip operations, as claimed. 
\end{proof}

%%%%%%%%%%%%%%%%%%%%%%%%%%%%%%%%%%%%%%%%%%%%%%%%%%%%%%%%%%%%%%%%%%%%%%%%

We conjecture that the same linear lower bound applies to the 2D
(disconnected) model of Snipperclips as well:

\begin{conjecture} \label{conj:2D disconnected}
  For $M$ a positive integer, consider the collection of all possible 2D ``comb'' tools consisting of a
  $1 \times M$ rectangle with $n$ disjoint $1 \times t_i$ ``teeth''
  attached above it (by its side of length $t_i$), where each tooth has integer coordinates and $1 \le t_i \le M$ so that the construction fits a $2\times M$ rectangle.
  For all positive integers $n$ and all $\epsilon \in (0,1)$, for all sufficiently large $M$, almost all such tools (a $1-\epsilon$ fraction of them) require
  $\Omega(n)$ snip operations to build,
  even with arbitrarily many tools (and thus with
  arbitrary undo, redo, and reset operations).
\end{conjecture}

\begin{figure}[ht]
	\centering
	\includegraphics{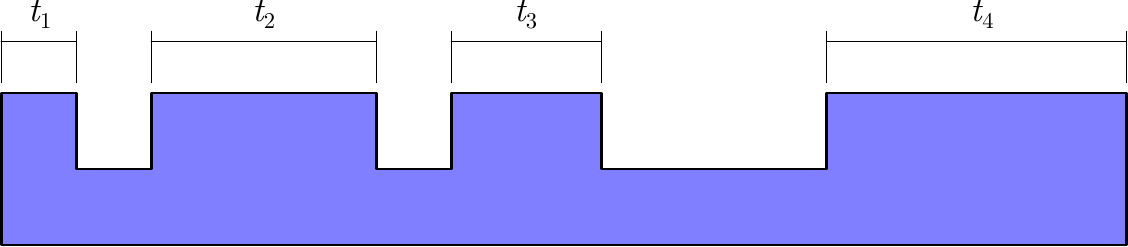}
	\caption{Illustration of Conjecture~\ref{conj:2D disconnected}. Note that because the teeth are disjoint and have integer coordinates, they are at least one unit apart.}
	\label{fig:conjecture}
\end{figure}

Unfortunately, a reduction from 2D Snipperclips to 1D Snipperclips remains
elusive.  A natural approach is to view a 2D tool $T$ as a set of 1D tools,
one for each direction that has perpendicular edges in~$T$.
But in this view, it is possible in a linear number of snips to construct
a 2D tool containing exponentially many 1D tools, by repeated generic
rotation and snipping of the tool by itself.
The information-theoretic argument of Theorem~\ref{thm:1D disconnected}
might still apply, but given the exponential number of tool choices in each
step, it would give only a logarithmic lower bound on the number of snips.
We can instead prove such a bound holds for \emph{all} shapes:

%%%%%%%%%%%%%%%%%%%%%%%%%%%%%%%%%%%%%%%%%%%%%%%%%%%%%%%%%%%%%%%%%%%%%%%%

\begin{theorem} \label{thm:disconnected log}
  Every tool shape with $n$ edges requires $\Omega(\log n)$ snip
  operations to build from initial shapes of $O(1)$ edges
  in the disconnected model.
  This result holds even with arbitrarily many tools (and thus with
  arbitrary undo, redo, and reset operations).
\end{theorem}

\begin{proof}
  Each snip operation involving two tools with $n_1$ and $n_2$ edges,
  respectively, produces a shape with at most $n_1+n_2$ edges.
  Thus, if we start with tools having $c=O(1)$ edges,
  then in $k$ snips we can produce a shape having at most $c^k$ edges,
  proving a lower bound of $k \geq \log_c n$.
\end{proof}

%%%%%%%%%%%%%%%%%%%%%%%%%%%%%%%%%%%%%%%%%%%%%%%%%%%%%%%%%%%%%%%%%%%%%%%%

\section{Making one shape with snips and resets}\label{sec_noundo}
%%%%%%%%%%%%%%%%%%%%%%%%%%%%%%%%%%%%
\subsection{Connected Model}\label{sec_conec}
%%%%%%%%%%%%%%%%%%%%%%%%%%%%%%%%%%%%
In the connected model, the shapes must remain connected. Whenever the snip operation would break a tool into multiple pieces, we can choose one piece to keep. In this model, we show that $O(n)$ snips suffice to create any polygonal shape of $n$ vertices.

\begin{theorem}\label{thm:connected}
We can cut one of the tools into any target polygonal domain $P_1$ of $n$ vertices using $O(n)$ snip operations (and no reset or undo operations) in the connected model.
\end{theorem}
\begin{proof}

The idea is that we can shape $\mathcal{T}_2$ into a very narrow triangle, a \emph{needle}, and use that to cut along the edges of the target shape $P_1$. Whenever a snip disconnects the shape, we simply keep the one containing the target shape. Initially, we start with a long needle to cut the long edges of $\mathcal{T}_2$ and we gradually shrink the needle to cut the smaller edges.

Let $\alpha$ and $h$ be two small numbers to be determined.
Our needle will be an isosceles triangle, with the two equal-length edges making an angle of $\alpha$ and the base edge with length at most $h$.
We refer to the \emph{length} of the needle as the length of the equal-length edges.
We will choose $\alpha$ small enough so that (i) the needle can fit into all reflex vertices, and we choose $h$ small enough so that, (ii) when placed on an edge of the target polygon, the needle does not intersect a non-adjacent edge.

\begin{figure}[ht]
\centering
\includegraphics{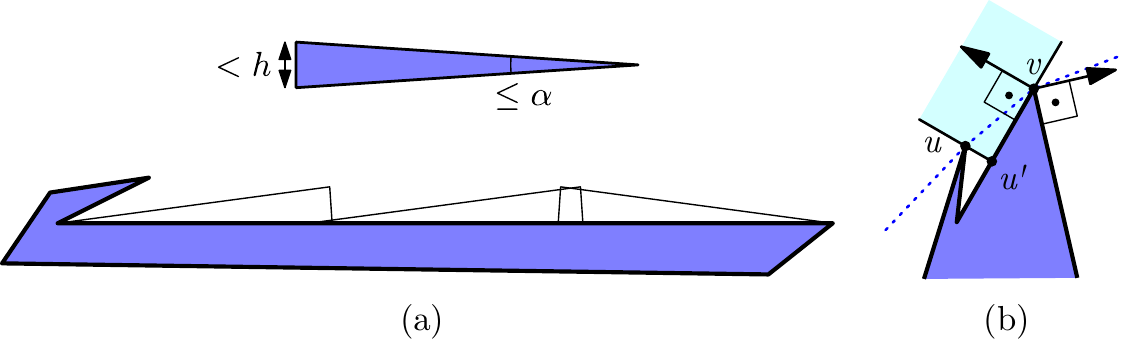}
\caption{(a) The needle is an isosceles triangle with apex at most $\alpha$ and a base edge of length at most $h$. The equal-length edges have length at most $1$ so that the whole triangle can fit inside a tool. (b) Dashed blue lines denote $\text{Ch}(P_1)$. The choice of $h$ guarantees that there is a segment of length $\ge h$ contained on the boundary of $\mathcal{T}_1$ that can be used to shrink the needle $\mathcal{T}_2$.}
\label{fig:needle}
\end{figure}

Refer to Figure~\ref{fig:needle}.
Let $v$ be an arbitrary vertex on the convex hull of $P_1$, denoted $\text{Ch}(P_1)$, and let $e_1$ and $e_2$ be its incident edges.
By the definition of convex hull, at least one edge in $\{e_1, e_2\}$ has the property that its normal vector at $v$ is outside of $\text{Ch}(P_1)$.
Without loss of generality, let that be $e_1$ and let $u$ be the vertex of $P_1$ whose orthogonal projection $u'$ on $e_1$ is closest to $v$ and $u$ lies on the closed half-plane defined by the supporting line of $e_1$ containing the normal vector.
Note that $u$ might be also in the convex hull and then $u = u'$.
We first make $\mathcal{T}_2$ into a needle of length $1/2$ using 2 snips. 
Fix a rigid transformation of $P_1$ so that it is entirely contained in $\mathcal{T}_1$.
We no longer move $\mathcal{T}_1$.
Use the needle to cut off a $90^\circ$ wedge at $u'$ containing the segment $u' v$ on its boundary and so that we do not cut off any point in the interior of $P_1$.
This is done with at most $4$ snips due to the length of the needle.

Now we group all edges of $P_1$ into sets based on their length. Let $\edgeset$ denote the full set of edges defining $P_1$ and let $\edgeset_i$, for $0 \leq i$, be the set of edges whose length is between $2^{-i-1}$ and $2^{-i}$. To cut along the edges of $\edgeset_i$, we use a needle where the equal-length edges have length $2^{-i-2}$. Such a needle can cut each edge in $\edgeset_i$ using at most four snips; see Figure~\ref{fig:needle}~(a). 
For an edge $e$, its nearest other features of $P_1$ are its two adjacent edges, the vertices closest to the edge, and the edges closest to its endpoints. 
We avoid cutting into the adjacent edges by placing the tip of the needle at the vertex when cutting near a vertex.
By Properties~(i)--(ii), we can make $e$ an edge of $\mathcal{T}_1$ without removing any point in the interior of $P_1$.

By making the cuts along the edges in the sets $\edgeset_i$ in increasing order of $i$ the needle has to only shrink, which is easily done by using the segment $u'v$ in the perimeter of $\mathcal{T}_1$ to shorten the needle by placing the short edge of the needle parallel to $u'v$.
This is possible as long as (iii) $h<\|u'v\|$ where $\|.\|$ denotes Euclidean norm. 
We are now ready to set $\alpha$ and $h$.
Property (i) is achieved if $\alpha$ is smaller than every external angle in $P_1$.
Property (ii) is achieved if $h$ is smaller than the shortest distance between an edge and a nonincident vertex. 
We also have that the length of the initial needle is $1/2$ and thus $\sin(\alpha/2) \le h$ using the law of cosines.

Recall that making the initial needle requires two snips, cutting each edge requires at most four snips and hence $O(n)$ snips in total, and reducing the needle length requires one snip per nonempty set $\edgeset_i$ of which there are at most $O(n)$. Thus, in total the required number of snips is $O(n)$. 
\end{proof}

%%%%%%%%%%%%%%%%%%%%%%%%%%%%%%%%%%%%
\subsection{Disconnected Model}\label{sec_disconec}
%%%%%%%%%%%%%%%%%%%%%%%%%%%%%%%%%%%%

We now consider the disconnected model. Recall that in this model we allow the tools to become disconnected. That is, when a snip would disconnect the tool, we keep all pieces. This is the actual version implemented in the game. Unfortunately, the method in the prior section will not work here. The first issue is that our tool must now remove the full area of the unwanted space rather than relying on separated components disappearing. The second issue is that we may cut up the boundary of our target in such a way that we can no longer ensure we have an exterior edge of sufficient size to efficiently trim our needle into the next needed shape. To solve these problems we end up using $O(n^2)$ snips and the reset operation which was not used in the previous section. The new algorithm works in phases where we only tackle an L-shaped portion of the shape at a time. This allows us to keep a solid square in the lower right which is sufficiently large to create the tools we need to carve out the desired shape. It also ensures that we can isolate the tool which we are using to carve the target region of the current phase. Thus each phase bounds how far into the target we must reach and ensures we have a block with which to alter our carving tool, allowing methods similar to those in Section~\ref{sec_conec} to complete each phase. We now give a formal description and proof of correctness.

In order to carve out a target shape $P_1$, we virtually fix a location of $P_1$ inside $\mathcal{T}_1$, pick a corner $c$ of $\mathcal{T}_1$ (say, the lower right one) and consider the set of distances $d_1, \ldots, d_{n'}$ from each of the vertices in the fixed location of the target shape $P_1$ to $c$ in decreasing order under the $L_\infty$-metric. For simplicity assume that all distances are distinct, and thus $n' = n$ (this can be achieved with symbolic perturbation). We refer to the part of $\mathcal{T}_1$ not in $P_1$, i.e., $\mathcal{T}_1 \setminus P_1$, as the \emph{free-space}. We will remove the free-space in $n$ steps, where in each step $i$ we remove the free-space from an $L$-shaped region $Q_i$ that is the intersection of $\mathcal{T}_1$ and an annulus formed by removing the $L_\infty$-ball of radius $d_i$ from the $L_\infty$-ball of radius $d_{i-1}$ centered at $c$. We argue that in each step we will need $O(n)$ snips and resets, thus creating the target shape in $O(n^2)$ operations. Our inductive step is given in the following lemma.

\begin{lemma}\label{lem:cut-L-shape}
The free-space in region $Q_i$ can be removed in $O(n)$ snip and reset operations provided that $\bigcup_{j \geq i} Q_j$ is a square in $\mathcal{T}_1$.
\end{lemma}
\begin{proof}
\begin{figure}[t]
\centering
\includegraphics{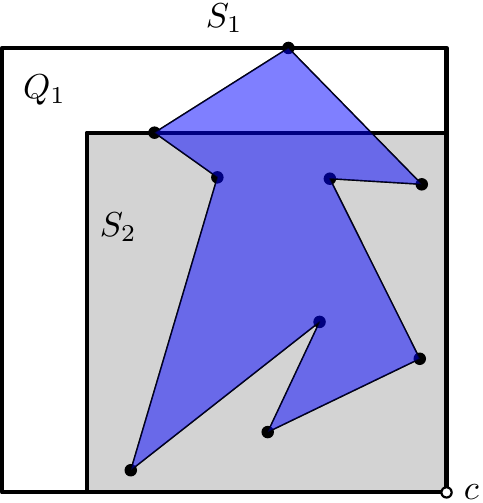}
\caption{The squares $S_1$ and $S_2$ along with L-shaped region $Q_1$ and corner $c$.}
\label{fig:definitions}
\end{figure}

Let $S_i$ be the bounding square containing $Q_i$ (see Figure~\ref{fig:definitions}) and let $F_i$ be the set of faces created when removing the boundary edges of the target shape from $Q_i$. By definition all vertices of the target shape on $Q_i$ must be on its inner or outer $L$-shaped boundary and all boundary segments must fully traverse $Q_i$, i.e., they cannot have an endpoint inside $Q_i$. It then follows that the set $F_i$ of faces consists of $O(n)$ constant complexity pieces. Now triangulate all faces of $F_i$ and let $T_i$ denote the resulting set of triangles (Figure~\ref{fig:L-shape}). Note that our aim is to remove some of the triangles of $T_i$. We will show that we can remove any triangle that fits in $S_i\setminus S_{i+1}$ with a constant number of cuts.

\begin{figure}[th]
\centering
\includegraphics{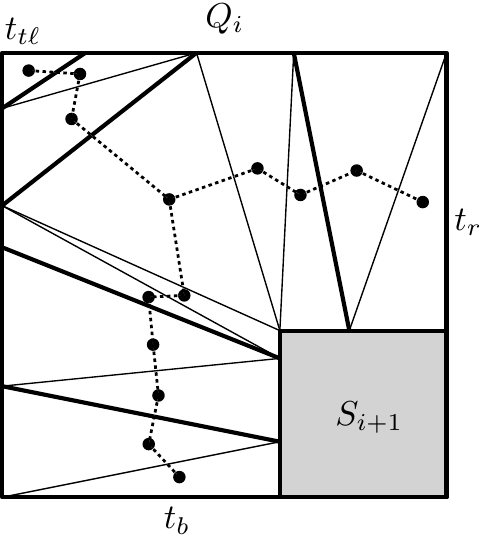}
\caption{An L-shaped region $Q_i$, the edges of the target shape that cross it (thick edges) define $F_i$. We further triangulate each face (thin edges), and consider the corresponding dual graph (dotted edges).}
\label{fig:L-shape}
\end{figure}

For simplicity in the exposition we first consider the case in which \textbf{\boldmath $S_{i+1}$ is large.} That is, the side length of $S_{i+1}$ is at least half the side length of $S_i$. Consider a triangle $t \in F_i$ that needs to be removed. To create a cutting tool move $\mathcal{T}_2$ so that its only overlap with $\mathcal{T}_1$ is $S_i$. Let $S'_i$ denote the area in $\mathcal{T}_2$ corresponding to $S_i$ and let $t'$ be the projection of $t$ on $\mathcal{T}_2$. Our goal will be to remove $S'_i \backslash t'$ from $\mathcal{T}_2$ without affecting $t'$. Note that we can create a cut where only $S'_i$ overlaps $\mathcal{T}_1$ in $S_i$, so the shape of $\mathcal{T}_2 \backslash S'_i$ does not influence the cut (Figure~\ref{fig:triangleprojection}). That means we do not have to cut it away and we do not need to worry about cutting part of it while creating a cutting tool within $S'_i$.

\begin{figure}[th]
\centering
\includegraphics{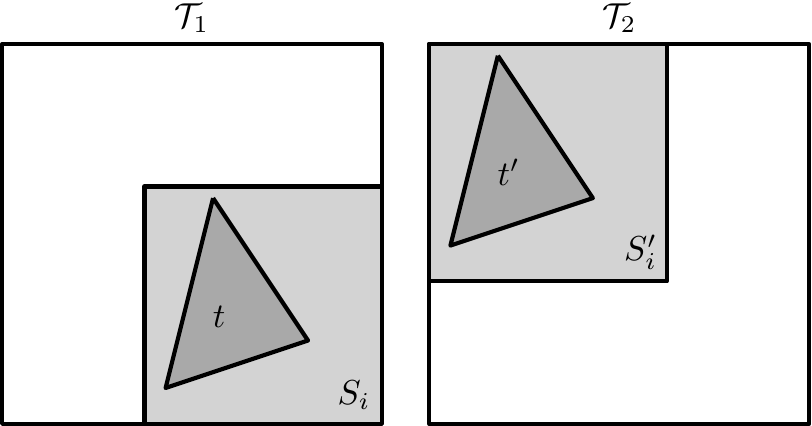}
\caption{A triangle $t$ in $S_i$ is cut out of $\mathcal{T}_2$ at $t'$.}
\label{fig:triangleprojection}
\end{figure}

Consider the halfspace $H$ defined by one of the bounding lines $\ell$ of $t'$ that does not contain $t'$. We can remove $H \cap S'_{i}$ by rotating $\mathcal{T}_1$ so that one of the sides of $\mathcal{T}_1$ along which $S_{i+1}$ is situated aligns with $\ell$ and repeatedly snip with $S_{i+1}$ in a grid-pattern as shown in Figure~\ref{fig:gridremoval}. Because $S_{i+1}$ is large compared to $S'_i$ we can remove $H \cap S'_{i}$ in $O(1)$ snips. We then apply the same procedure for the other two halfspaces that should be removed to obtain the cutting tool for $t$.
This means that, under the assumption that $S_{i+1}$ is large, each triangle can be removed in $O(1)$ snips. Since there are $O(n)$ triangles in $S_i$, the linear bound holds. 

\begin{figure}[th]
\centering
\includegraphics{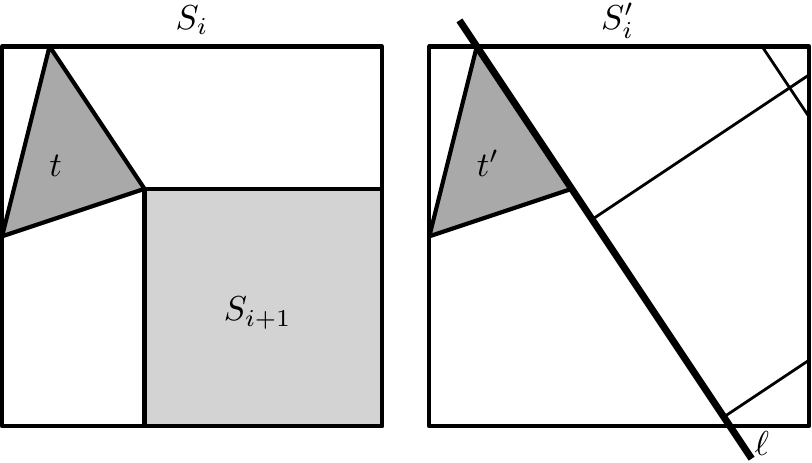}
\caption{If $S_{i+1}$ is large, we can use it to carve out any desired shape in $\mathcal{T}_2$ with $O(1)$ snips.}
\label{fig:gridremoval}
\end{figure}

It remains to consider the case in which \textbf{\boldmath $S_{i+1}$ is small.} That is, the side length of $S_{i+1}$ is less than half that of $S_i$, and potentially much smaller. Although the main idea is the same, we need to remove the triangles in order, and use portions of $Q_i$ that are still solid to create the cutting tools.

Let $G_i$ be the dual graph of $T_i$. This graph is a tree with at most three leaves. Two leaves correspond to the unique triangles $t_b$ and $t_r$ that share an edge with the lower and right boundary of $Q_i$ respectively and the third exists only if the top-left corner of $Q_i$ is contained in a single triangle $t_{t\ell}$, that is, there is at least one segment contained in $Q_i$ that connects the  top and left boundaries; see Figure~\ref{fig:L-shape}. Finally, we change the coordinate system so that $c$ is the origin, and $S_i$ is a unit square (note that the vertices of this square are $(-1,1)$, $(-1,0)$, $(0,1)$, and $c=(0,0)$).

We process the triangles in the following order. We first process the \emph{cross-triangles}, triangles with one endpoint on the left boundary and one on the top boundary (if any exist), starting from $t_{t\ell}$ following $G_i$ until we find a triangle that has degree three in $G_i$ which we do not process yet. The remaining \emph{fan-triangles} form a path in $G_i$ which we process from $t_b$ to $t_r$.

\textbf{Cross-triangles.}
Recall that, by the way in which we nest regions $Q_i$, there cannot be vertices to the right or below $S_i$. In particular, cross-triangles have all three vertices in the top and left boundaries of $Q_i$. Hence, while we have some cross-triangle that has not been processed, the triangle of vertices $(-1,0)$, $(0,1)$ and $c$ must be present in $\mathcal{T}_1$. This triangle has half the area of $Q_i$ and can be used to create cutting pieces in the same way as in the case where we assumed $S_{i+1}$ is large. Thus, we conclude that any cross-triangle of $Q_i$ can be removed from $\mathcal{T}_1$ with $O(1)$ snips. 

\textbf{Fan-triangles.}
We now process the fan-triangles in the path from $t_b$ to $t_r$ in $G_i$. We treat this sequence in two phases. First consider the triangles that have at least one vertex on the left edge of $S_i$ (that is, we process triangles up to and including the triangle that has degree three in $G_i$ if it exists). Consider the triangle $t$ of vertices $(0,1)$, $(0,3/4)$, and $(-1/4,3/4)$ (see Figure~\ref{fig:fan-triangles}). This triangle has $1/32$ of the total area of $S_i$. It is also still fully part of $\mathcal{T}_1$ until we cut out the triangle of degree 3. That is, every cross-triangle that was cut is above the diagonal from (-1,0) to (0,1) and any fan-triangle that has at least one vertex on the left edge of $S_i$ and has degree two in $G_i$ is below the line from (-1,1) to (0,1/2) (technically, below the line from (-1,1) to the top-right corner of $S_{i+1}$, but the higher line suffices for our purposes). So we can use this triangle $t$ as a cutting tool to create the desired triangle in $\mathcal{T}_2$ to cut out any undesired fan-triangles up to and including the triangle of degree 3.

\begin{figure}
\centering
\includegraphics{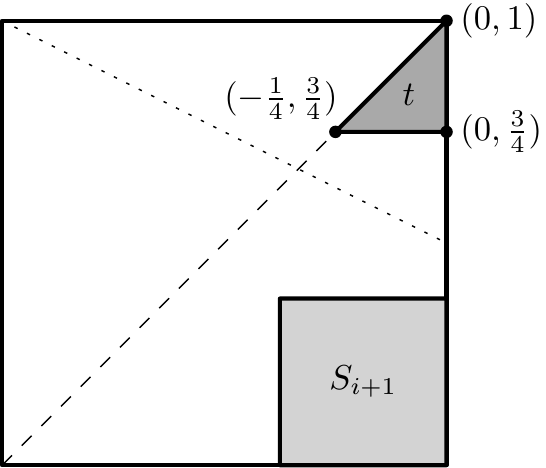}
\caption{The triangle used to cut out the fan-triangles. Cut cross-triangles are above the dashed line and cut fan-triangles are below the dotted line.}
\label{fig:fan-triangles}
\end{figure}

The remaining triangles have their vertices in the upper edge of $S_{i}$ and on the upper or left edge of $S_{i+1}$. In this case we must be more careful as we cannot guarantee the existence of a large square in $\mathcal{T}_1$. However, we do not have to clear the entire space $S_i'$ any longer. Instead it suffices to clear a much smaller area.

\begin{figure}[th]
\centering
\includegraphics{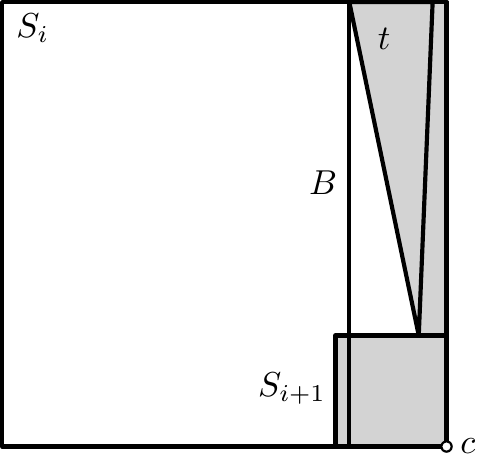}
\caption{The solid areas (grey) and bounding box $B$ when cutting fan-triangles with no vertices on the left boundary of $S_i$.}
\label{fig:bboxcut}
\end{figure}

Let $t$ denote the next triangle to be removed and let $B$ denote the bounding box of $t$ and $c$ (see Figure~\ref{fig:bboxcut}). As before consider moving $\mathcal{T}_2$ so that the only overlap with $\mathcal{T}_1$ is $B$, let $B'$ denote this area in $\mathcal{T}_2$ and $t'$ the projection of $t$ onto $B'$. To create a cutting tool we need only remove the area $B' \backslash t'$.

As before, we look for a region in $\mathcal{T}_1$ that has roughly the area of $B$ to use for carving the desired shape in $\mathcal{T}_2$. Let $w$ be the width of $B$. Also, let $h'$ be the height of $S_{i+1}$. Note that the height of $B$ is $1$, and since $S_{i+1}$ is small, we have $h'<1/2$. By construction of the bounding box, one of the vertices of $t$ will have $x$-coordinate equal to $-w$; let $q$ denote this vertex. The $y$-coordinate $y_q$ of $q$ is either $1$ or $h'$ as it must be on the upper edge of $S_i$ or on the upper boundary of $S_{i+1}$---if $t$ has vertices on the left boundary of $S_{i+1}$, then there is a vertex on the upper boundary of $S_i$ with lower $x$-coordinate. Now consider the triangle with vertices $(0,1), (0,h'), q$. This triangle has height at least $1-h'>1/2$ and width $w$, and thus its area is at least $1/4$ of the area of $B$. As in the previous cases, we use this triangle to create a cutting tool from $\mathcal{T}_2$ to remove triangle $t$ from~$\mathcal{T}_1$.

Thus, it follows that all free-space triangles can be removed with a cutting tool that is constructed from $\mathcal{T}_2$ in $O(1)$ snip and reset operations, hence we can clear $Q_i$ of free-space in total $O(n)$ operations.
\end{proof}

Because there are at most $n$ distinct distances, we repeat this procedure at most $n$ times, giving us the desired result.

\begin{theorem}
We can cut one of the tools into any target polygonal domain $P_1$ of $n$ vertices using $O(n^2)$ snip and reset operations in the disconnected model.
\end{theorem}

\section{Adding the undo operation}\label{sec_undo}
We now consider a more powerful model in which we can  {\em undo} the $k$ latest operations performed on either of the tools.
More formally, each snip or reset operation will change the current shape of one of the two tools (if a snip or reset operation does not change the shape of either tool, we can ignore it). Given a sequence of such operations, consider the subsequence $o_1, \ldots o_m$ of operations that have changed the shape of the first tool. Also, let $P_1^{(i)}$ be the shape of the first tool after $o_i$ has been executed. The $k$-{\em undo} operation on the first tool replaces the current shape with $P_1^{(m-k)}$. The $k$-undo operation on the second tool is defined analogously.

In this section we show that the $k$-undo operation is very powerful, and allows us to do much more than we can do without it. In particular, we can transform two tools into any two target polygonal domains in both the connected and disconnected model. This statement holds even if we force $k$ to be equal to 1.

\subsection{Connected Model}
We first consider the connected model. The general idea in this case is that we first construct the target shape in one of the two tools. In order to construct the target shape into the second tool, we repeatedly create a needle in the first tool, cut a part of the second tool, and perform an undo operation to return the first tool to its target shape.

\begin{theorem}
We can cut two tools $\mathcal{T}_1$ and $\mathcal{T}_2$ into any two target polygonal domains $P_1$ and $P_2$ of $n$ and $m$ vertices respectively using $O(n+m)$ snip, reset and $1$-undo operations in the connected model.
\end{theorem}
\begin{proof}
Let $e_1$ be the longest edge of $P_1$ not on the boundary of the unit square  and $e_2$ be the longest edge of $P_2$ not on the boundary of the unit square.  Without loss of generality, we assume that $e_1$ is longer than $e_2$.
We apply Theorem~\ref{thm:connected} to cut $\mathcal{T}_1$ into $P_1$. To create $P_2$ we will use a needle to cut along edges as in Theorem~\ref{thm:connected}. Each needle will be cut along $e_1$ using a small construction along $e_2$. We will ensure the needle can have varying sizes, so we can cut along each edge in $O(1)$ cuts. We also guarantee that the needle can be created from $P_1$ in a single cut, so we can easily undo the operation.

\begin{figure}[ht]
\centering
\includegraphics{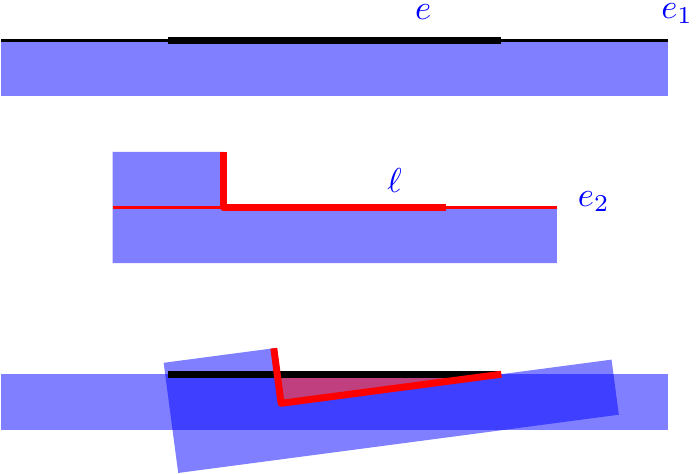}
\caption{We can use $e_1$, $e_2$ and a small added edge on $e_2$ to create a needle in $\mathcal{T}_1$ that can be used to create $P_2$ in $\mathcal{T}_2$. The needle is indicated in purple.}
\label{fig:1-undo-connected-needle}
\end{figure}

We first explain how to create the needle, as also illustrated in Figure~\ref{fig:1-undo-connected-needle}. We create the needle from a segment $e$ of $P_1$, which is a subsegment of $e_1$ that is half the length of $e_1$ but centered at its center. The cutting tool will consist of a subsegment $\ell$ of $e_2$ and an edge perpendicular to it, creating a $90^\circ$ angle in the freespace. The segment $\ell$ is also half the length of $e_2$ and centered at its center. This is to ensure that there is a constant size rectangle above and below $e$ and $\ell$ that does not contain edges or vertices of $P_1$ or $P_2$. Now to cut a needle from along $e$, assume that $e$ is horizontal with freespace above it and that the edge perpendicular to $\ell$ is on its left endpoint oriented upward. Now place the right endpoint of $\ell$ on the right endpoint of $e$ and rotate $\ell$ counterclockwise around the right endpoint by an arbitrarily small angle so that the right angle is in the interior of $P_1$, just below $e$. This cut will disconnect a needle from the rest of $P_1$ with a length proportional to $\ell$. By moving $\ell$ higher before cutting we can create shorter needles.

\begin{figure}
\centering
\includegraphics{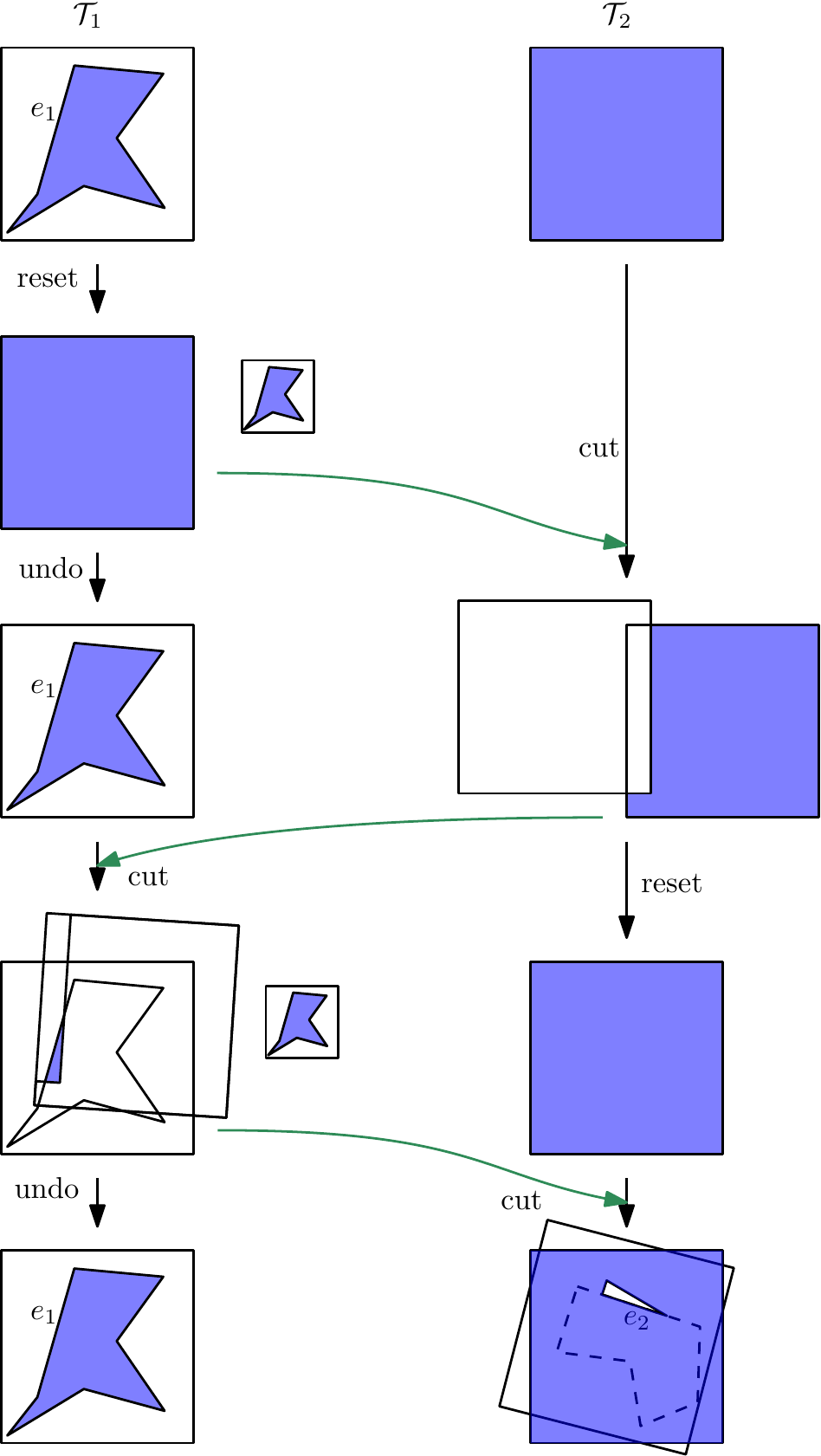}
\caption{Steps illustrating the creation of a freespace triangle above $e_2$ that can be used to created needles along $e_1$. Small squares indicate the shape that is on the undo stack (omitted when not used later).}
\label{fig:1-undo-connected-start}
\end{figure}

For this cutting process to work, the triangle created by $\ell$ and the perpendicular edge must be empty. So this will be the first piece we remove from $\mathcal{T}_2$ in the process of creating $P_2$. How to do this is illustrated in Figure~\ref{fig:1-undo-connected-start} and described next. We first reset $\mathcal{T}_1$ and cut a long narrow rectangle out of the top left corner of $\mathcal{T}_2$. This gives us a long vertical edge and a shorter horizontal edge perpendicular to it. We use this structure to create a narrow triangle along $e_1$ as described above. This needle is then aligned with $e_2$ and cuts out a narrow triangle above $e_2$ so that an edge perpendicular to $e_2$ is created that is sufficiently far from the endpoint of $e_2$.

The remainder of the process follows that of Theorem~\ref{thm:connected} where we use needles of a specific length to cut edges proportional to that length. The one exception is $e_2$, which is cut last. Note that unlike in Theorem~\ref{thm:connected}, the order in which we cut the edges is no longer relevant, since we can cut the needle to the size required for the current edge, cut that edge, and then return the needle to its original length using a 1-undo operation. This guarantees that the perpendicular edge required stays attached to the main shape and is removed only when no more needles need to be created.
\end{proof}

\subsection{Disconnected Model}
Finally, we focus our attention on the disconnected model with undo operations. We show that allowing undo operations reduces the upper bound on the number of operations required to cut one target shape out of one tool. In fact, we can cut any two target shapes out of the two tools, but the number of operations needed for this depends on the size of the undo stack.

\begin{theorem}
\label{thm:P_1_undo_1}
We can cut one of the tools into any target polygonal domain $P_1$ of $n$ vertices using $O(n)$ snip, reset and $1$-undo operations in the disconnected model.
\end{theorem}

\begin{proof}
We first triangulate the free-space $\mathcal{T}_1 \backslash P_1$. Then, we remove each triangle $t$ by making a congruent triangle $t'$ in $\mathcal{T}_2$.
Each time we create a triangle $t'$ in $\mathcal{T}_2$ we first reset $\mathcal{T}_1$ and $\mathcal{T}_2$. Then, we can remove $\mathcal{T}_2 \backslash t'$ using $\mathcal{T}_1$ with a constant number of snips. Since we only apply one operation on $\mathcal{T}_1$, we can use an undo operation to restore $\mathcal{T}_1$ to its previous shape, which is  the partially constructed shape towards the target shape $P_1$. Next, we can cut the triangle $t$ in $\mathcal{T}_1$ using the congruent triangle $t'$ in $\mathcal{T}_2$. Thus, we use $O(1)$ snip, reset and 1-undo operations.
We apply this process for each triangle in the free-space. Hence, since the triangulation has linear complexity, we can remove the free-space with $O(n)$ operations in total.
\end{proof}

Next, we show that we can cut the two tools into any two target shapes using only snip, reset and 1-undo operations.

\begin{theorem}
We can cut two tools $\mathcal{T}_1$ and $\mathcal{T}_2$ into any two target polygonal domains $P_1$ and $P_2$ using a finite number of snip, reset and $1$-undo operations in the disconnected model.
\end{theorem}

\begin{proof}
We apply Theorem~\ref{thm:P_1_undo_1} to cut $\mathcal{T}_1$ into $P_1$. Then, the idea is that we can shape $P_1$ into a very narrow triangle, {\em a needle}, by using one snip  operation, and use the needle to cut all the free-space $\mathcal{T}_2 \backslash P_2$. After we get $P_2$, we can perform a 1-undo operation to restore $\mathcal{T}_1$ to $P_1$.

Let $\alpha$ be the smallest angle between any two adjacent edges of $P_2$, $d$ be the length of the shortest edge of $P_2$, and $h$ be the shortest distance between any vertex and a non-adjacent edge of $P_2$. These values will define how small the needle we create needs to be. Let $\ell_1$ be the vertical line touching the leftmost vertex of $P_1$. Since there may be multiple such vertices, let $p$ be the bottommost vertex of $P_1$ on $\ell_1$. Let $\ell_2$ be the vertical line touching the first vertex on the right side of $\ell_1$ in $P_1$.
We first reset $\mathcal{T}_2$ to a unit square. We align the left edge  $L_2$ of $\mathcal{T}_2$ with $\ell_1$ such that $\mathcal{T}_2$ fully covers $P_1$. Then, we shift $\mathcal{T}_2$ a little bit to the right such that $L_2$ is between $\ell_1$ and the bisector of $\ell_1$ and $\ell_2$, and the length of the bottommost edge of $P_1$ between $\ell_1$ and $L_2$ is less than $d/2$. We cut $P_1$ with $\mathcal{T}_2$ so that we have a set $T$ of triangles (or quadrangles) left in $\mathcal{T}_1$ (see Figure~\ref{fig:1-undo}).

\begin{figure}[ht]
\centering
\scalebox{1}{\includegraphics[page=1]{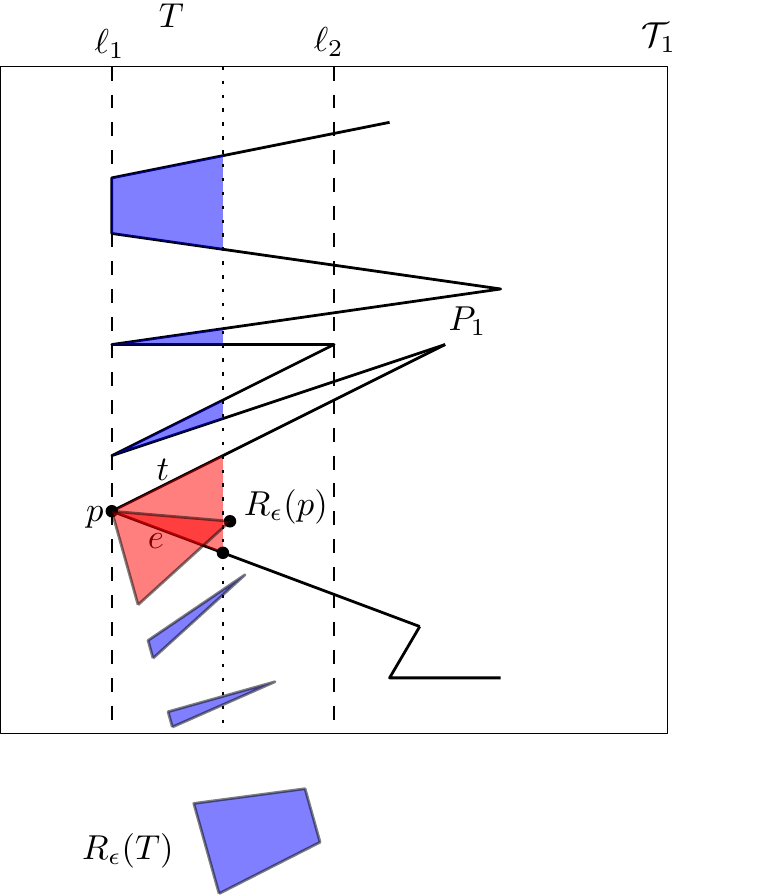}}
\caption{The figure shows the set $T$ of the triangles and quadrangles (with filled colors) after cutting $P_1$ with the unit square $\mathcal{T}_2$ and the set $R_{\epsilon}(T)$ obtained by rotating $T$ $180^{\circ}$ around the midpoint of $e$ and then rotating a small angle $\epsilon$ counterclockwise around $p$.}
\label{fig:1-undo}
\end{figure}

\begin{figure}[ht]
\centering
\scalebox{1}{\includegraphics[page=2]{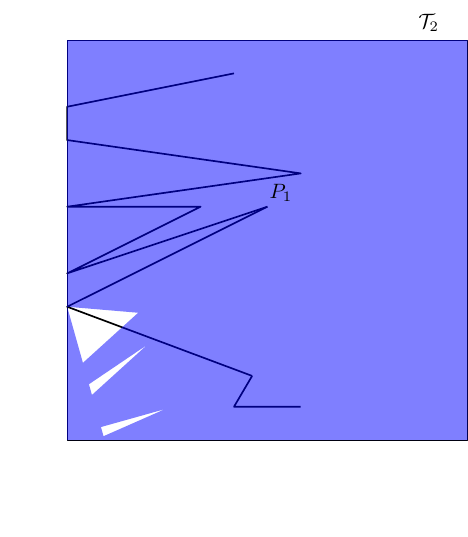}}
\caption{The figure shows $\mathcal{T}_2$ after removing $R_{\epsilon}(T)$ and the part of the boundary of $P_1$.}
\label{fig:1-undo_2}
\end{figure}

Let $e$ be the bottommost edge of $T$ and let $t$ be the bottommost object of $T$. Let $R$ be the function that rotates the input by $180^{\circ}$ around the midpoint of $e$, i.e., $R(T)$ is the set of triangles (or quadrangles) obtained by rotating $T$ $180^{\circ}$ around the midpoint of $e$, and $R(t)$ is the triangle obtained by rotating $t$ in the same manner. Notice that the intersection of $R(T)$ and $T$ is only $e$.
Let $R_{\epsilon}$ be the function that rotates the input by $180^{\circ}$ around the midpoint of $e$ and then rotates it by a small angle $\epsilon$ counterclockwise around $p$ of $T$.
We pick a small $\epsilon < \alpha/2$ such that no triangle in $R_{\epsilon}(T)$ crosses $\ell_2$, only $R_{\epsilon}(t)$ intersects with $t$, and the distance between $R_{\epsilon}(p)$ and $e$ is less than $h/2$.
We shift $\mathcal{T}_2$ back to the left such that $L_2$ is on $\ell_1$. Then, we perform the rotation $R_{\epsilon}$ on $T$ and cut $\mathcal{T}_2$ with $R_{\epsilon}(T)$.
After this cut, we perform an undo operation to restore $\mathcal{T}_1$ to $P_1$ and rotate $P_1$  back to its starting orientation. Finally, we cut $P_1$ with $\mathcal{T}_2$ to obtain the needle (see Figure~\ref{fig:1-undo_2}).

We argue why the final cut indeed leaves only the needle. Since $\mathcal{T}_2$ almost covers $P_1$ except for the missing part $R_{\epsilon}(T)$,
it is essential to show that the intersection of $R_{\epsilon}(T)$ and $P_1$ is the needle. Since $R(T)$ lies between $\ell_1$ and the bisector of $\ell_1$ and $\ell_2$, there exists a small $\epsilon$ such that $R_{\epsilon}(T)$ lies between $\ell_1$ and $\ell_2$. In addition, $e$ is the bottommost edge of $T$, so there cannot be any intersection of $R_{\epsilon}(T)$ and $P_1$ below $e$.  The intersection of $P_1$ and $R(T)$ is $e$ and all the triangles in $R(T)$ are below $e$, so we can rotate them by a small angle $\epsilon$ around $p$ so that only one vertex $R_{\epsilon}(p)$ in $R_{\epsilon}(T)$ lies above $e$ (see Figure~\ref{fig:1-undo}). As one of the endpoints of the edge of $P_1$ that contains $e$ lies on or to the right side of $\ell_2$, the intersection of $P_1$ and $R_{\epsilon}(t)$ is a triangle.
In particular, the intersection of $R_{\epsilon}(T)$ and $P_1$ is a narrow triangle with a base length of at most $d/2$, height of at most $h/2$ and a small angle of at most $\alpha/2$.

After we obtain the needle, we reset $\mathcal{T}_2$ and use the needle to cut the free-space $\mathcal{T}_2 \backslash P_2$ in a finite number of snip operations, because the free-space is a compact object. Finally, we perform an undo operation to restore the needle to $P_1$, resulting in the two target polygonal domains.
\end{proof}

Finally, we show that if we are allowed to use a 2-undo operation instead of a 1-undo, the number of required operations reduces to linear in the complexity of the two target polygonal domains.

\begin{theorem}
We can cut two tools $\mathcal{T}_1$ and $\mathcal{T}_2$ into any two target polygonal domains $P_1$ and $P_2$ using $O(n+m)$ snip, reset and $2$-undo operations in the disconnected model.
\end{theorem}

\begin{proof}
We apply Theorem~\ref{thm:P_1_undo_1} to cut $\mathcal{T}_1$ into $P_1$. Then, we define a cover of the free-space $\mathcal{T}_2 \backslash P_2$ with only small right triangles. We remove each right triangle $t$ by making a congruent triangle $t'$ in $\mathcal{T}_1$ by performing at most two operations on $P_1$, so we can get the target shape $P_2$ and restore $\mathcal{T}_1$ to $P_1$.

We first explain how to define the cover of the free-space with only right triangles. We start with any triangulation on the free-space $\mathcal{T}_2 \backslash P_2$. Then, we subdivide each triangle into a constant number of smaller triangles such that each smaller triangle fits in a $\frac{1}{2} \times \frac{1}{2}$ square. For each triangle, we draw a line segment from the vertex of the largest angle perpendicular to its opposite edge in order to split the triangle into two right triangles. Hence, there are $O(m)$ right triangles in the cover.

Next, we describe how to create the cutting tool in $\mathcal{T}_1$ (see Figure~\ref{fig:2-undo}). For each right triangle $t$ in the free-space, we first reset both $\mathcal{T}_1$ and $\mathcal{T}_2$ ($P_1$ and the partially constructed $P_2$ are stored at the top of their stacks).  We use the unit square $\mathcal{T}_2$ to cut the unit square $\mathcal{T}_1$ to get a triangle $t'$ congruent to $t$ at a corner of $\mathcal{T}_1$ ($P_1$ is stored at the second element of its stack). Note that there are other garbage components left in $\mathcal{T}_1$. Then, we translate $\mathcal{T}_1$ in such a way that only $t'$ overlaps $\mathcal{T}_2$, and cut $\mathcal{T}_2$ to make a square with a triangular hole (the partially constructed $P_2$ is at the second element of its stack). We perform an undo operation to restore $\mathcal{T}_1$ back to the unit square. The next step is to align the bounding unit square of $\mathcal{T}_1$ and $\mathcal{T}_2$, and cut $\mathcal{T}_1$ with $\mathcal{T}_2$ so that we get only $t'$ in $\mathcal{T}_1$. After we get the cutting tool $t'$, we perform two undo operations to restore $\mathcal{T}_2$ to the partially constructed $P_2$, and use $t'$ to remove $t$ from the free-space. Finally, we perform two undo operations to restore $\mathcal{T}_1$ to $P_1$. Overall, we use $O(1)$ snip, reset and undo operations to make some progress on $\mathcal{T}_2$ towards $P_2$ while maintaining $P_1$.

\begin{figure}
\centering
\scalebox{0.8}{\includegraphics{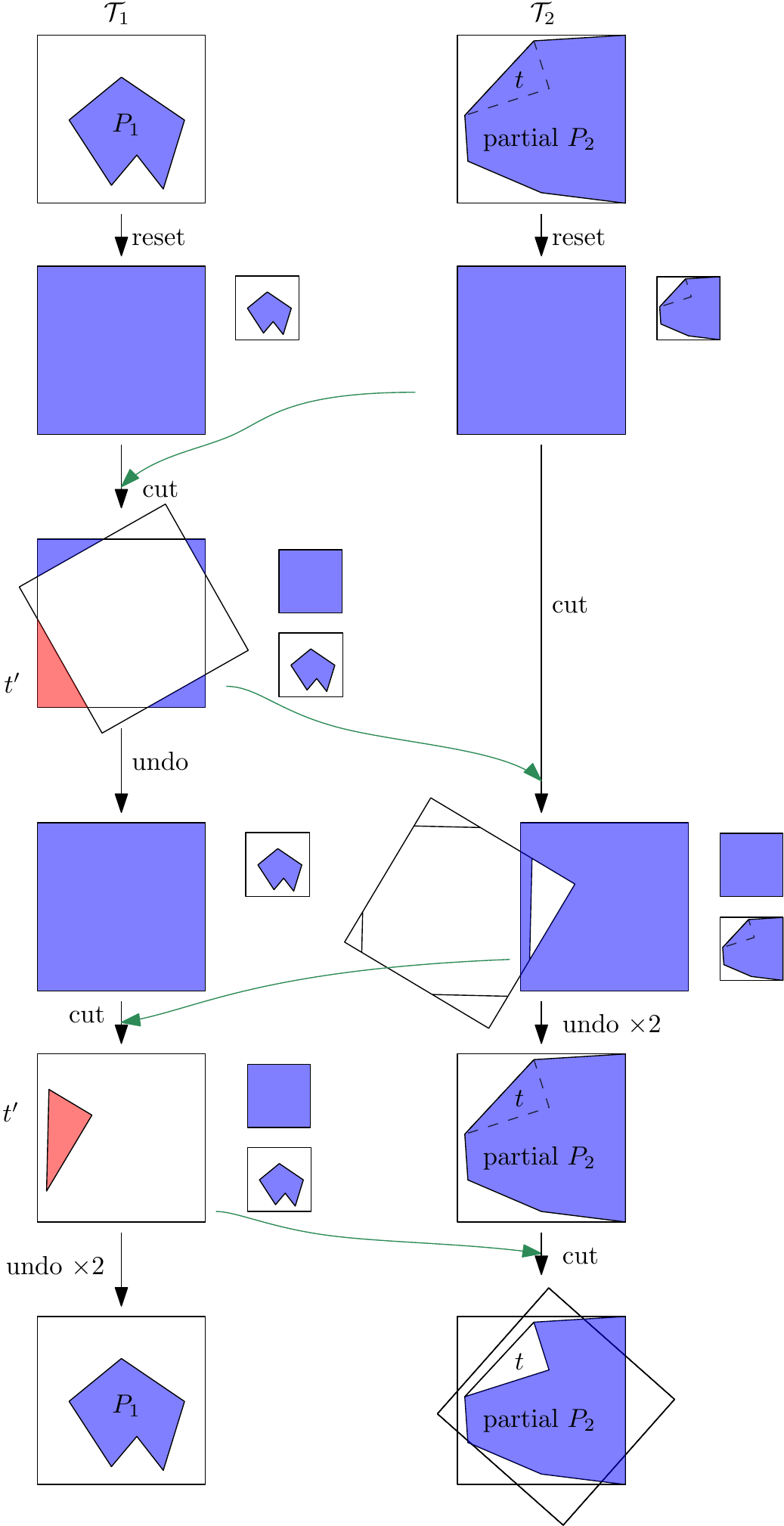}}
\caption{The figure shows how to remove a triangle $t$ in the partially constructed $P_2$ of $\mathcal{T}_2$ while maintaining $P_1$. Smaller squares indicate which shapes are on the undo stack.}
\label{fig:2-undo}
\end{figure}

We repeat the above process for each right triangle in the free-space, so we use $O(m)$ operations to carve out $P_2$. Including the $O(n)$ operations to carve out $P_1$, we use $O(n+m)$ operations in total.
\end{proof}

\section{Open Problems}
The natural open problem is to close the gap between our algorithms and the lower bound. Specifically, we are interested in a method that could extend our lower bound approach to the case in which you have the undo operation. We believe that without the undo operation there must exist a shape in the disconnected model that needs $\omega(n)$ operations to carve.

Our algorithms focus on worst-case bounds, but we also find the minimization problem interesting. Specifically, can we design an algorithm that cuts one (or two) target shapes with the fewest possible cuts? Is this problem NP-hard? If so, can we design an approximation algorithm?
Although it is not always possible to cut two tools simultaneously into the desired polygonal
shapes, it would be interesting to characterize when this is possible. Is the decision problem NP-hard?

It would also be interesting to consider the initial shape implemented in
the Snipperclips game (instead of the unit squares we used for simplicity),
namely, a unit square adjoined with half a unit-diameter disk.
This initial shape opens up the possibility of making curved target shapes
bounded by line segments and circular arcs of matching curvature.
Can all such shapes be made, and if so, by how many cuts? 

The stack size has a big impact in the capabilities of what we can do and on how fast can we do it. Additional tools can have a similar effect, since they can be used to {\em store} previous shapes. It would be interesting to explore if additional tools have the same impact as the undo operation or they actually allow more shapes to be constructed faster.

\section*{Acknowledgments}
This work was initiated at the 32nd Bellairs Winter Workshop on
Computational Geometry held January 2017 in Holetown, Barbados.
We thank the other participants of that workshop for providing a fun
and stimulating research environment.
We also thank Jason Ku for helpful discussions about (and games of)
Snipperclips.

\small
\bibliography{snipper}
\bibliographystyle{abbrv}

\end{document}